\documentclass[twocolumn,floatfix,showpacs,pra]{revtex4}

\usepackage{amsmath}
\usepackage{amssymb}
\usepackage{amsthm}
\usepackage{amsfonts}
\usepackage{graphicx}
\usepackage{color}
\usepackage{hyperref}
\newtheorem{lem}{Lemma}

\newtheorem{theorem}{Theorem}
\theoremstyle{definition} \newtheorem{defn}{Definition}

\newcommand{\ket}[1]{| {#1} \rangle}

\begin{document}

\title{Logical operator tradeoff for local quantum codes}
\author{Jeongwan Haah and John Preskill}
\affiliation{Institute for Quantum Information and Matter, California Institute of Technology, Pasadena, California 91125, USA}
\date{2 July 2012}

\begin{abstract}
We study the structure of logical operators in local $D$-dimensional quantum codes, considering both subsystem codes with geometrically local gauge generators and codes defined by geometrically local commuting projectors. We show that if the code distance is $d$, then any logical operator can be supported on a set of specified geometry containing $\tilde d$ qubits, where $\tilde d d^{1/(D-1)} = O(n)$ and $n$ is the code length. Our results place limitations on partially self-correcting quantum memories, in which at least some logical operators are protected by energy barriers that grow with system size.
We also show that for any two-dimensional local commuting projector code there is a nontrivial logical ``string'' operator supported on a narrow strip, where the operator is only slightly entangling across any cut through the strip. 
\end{abstract}
\maketitle
\section{Introduction}

Geometrically local quantum codes provide intriguing models of quantum many-body physics, and also have potential applications to fault-tolerant quantum computation in systems with short-range interactions. There has been impressive recent progress in understanding the properties of such codes. 
Bravyi, Poulin, and Terhal \cite{BravyiPoulinTerhal2010Tradeoffs} showed that for codes defined by geometrically local commuting projectors in $D$ dimensions, the code length $n$, distance $d$ and number of encoded qubits $k$ are related by 
\[
kd^{2/(D-1)} = O(n).
\]
Bravyi and Terhal \cite{BravyiTerhal2008no-go} showed that
\[
d = O(n^{(D-1)/D})
\]
for subsystem codes with geometrically local gauge generators, and
Bravyi \cite{Bravyi2010Subsystem} showed that
\[
kd = O(n)
\]
for two-dimensional subsystem codes with geometrically local gauge generators. 

Bravyi and Terhal \cite{BravyiTerhal2008no-go}, and Kay and Colbeck \cite{KayColbeck2008Quantum}, also showed that no two-dimensional local stabilizer code can be a \emph{self-correcting quantum memory} --- if we regard the code as a system governed by a local Hamiltonian, the energy barrier protecting against logical errors is a constant independent of system size. A self-correcting memory based on a geometrically local stabilizer code is possible in four dimensions \cite{DennisKitaevLandahlEtAl2002Topological,AlickiHorodeckiHorodeckiEtAl2008thermal}, where the storage time increases sharply as the system size grows. In three dimensions there are codes  such that the energy barrier increases logarithmically with system size \cite{Haah2011Local,BravyiHaah2011Energy}, but where the storage time is bounded above by a constant independent of system size \cite{BravyiHaah2011Analytic}.


We 
address a related but somewhat different question. To illustrate the question, consider the three-dimensional toric code \cite{CastelnovoChamon2008Topological}, on a cubic lattice with linear size $L$. This code provides different degrees of protection against different types of errors. For example, we can arrange for the logical bit flip acting on the code space to have weight $L$ (\emph{i.e.}, to be supported on a set of $L$ qubits), while the logical phase flip has weight $L^2$. In that case, the energy barrier protecting against logical phase errors grows linearly with $L$, though the energy barrier protecting against bit flips is only a constant. We might say this system is \emph{partially self correcting}, meaning it has very robust physical protection against phase errors, but weaker protection against bit flips. 

We find limitations on partial self correction in two-dimensional local subsystem codes with local stabilizer generators; in particular the logical phase flip must have weight $O(L)$ if the logical bit flip has weight $\Omega(L)$. More generally, we study how the code distance $d$ constrains the weight of logical operators, for both local commuting projector codes and subsystem codes, finding that $d$ limits not just the weight of the lowest-weight logical operator but also the higher-weight logical operators. Let us say that two logical operators are \emph{equivalent} if they act in the same way on the protected system. Our result, which applies to both local subsystem codes and to local commuting projector codes in $D\ge 2$ dimensions, says that for any logical operator there is an equivalent logical operator with weight $\tilde d$ such that
\begin{equation}\label{eq:main-result}
\tilde d d^{1/(D-1)} = O(L^D)
\end{equation}
where $L$ is the linear size of the lattice. We call this result the tradeoff theorem for logical operators, since, \emph{e.g.}, increasing the weight of the lowest-weight logical operator reduces the upper bound on the weight of other logical operators. One immediate consequence is that, since $d \le \tilde d$,
\[
d = O(L^{D-1}),
\]
a result previously known for local subsystem codes but not for local commuting projector codes with $D\ge 3$. For $D=2$ the tradeoff becomes $d\tilde d = O(L^2)$, and hence $d=O(L)$. 

We also show that for any two-dimensional local commuting projector code there is a nontrivial logical ``string'' operator supported on a narrow strip (or on a narrow slab in higher dimensions), where the operator is only slightly entangling across any cut through the strip. However, we have not settled the question whether two-dimensional local commuting projector codes can be self correcting.

We review the theory of stabilizer codes and subsystem codes in Sec. II. In Sec. III we prove a ``Cleaning Lemma'' for subsystem codes previously stated by Bravyi \cite{Bravyi2010Subsystem}; our proof uses tools developed by Yoshida and Chuang \cite{YoshidaChuang2010Framework}, and may be of independent interest. We prove the tradeoff theorem for local subsystem codes in Sec. IV and for local commuting projector codes in Sec. V. In Sec. VI we show that any two-dimensional commuting projector code admits a nontrivial logical ``string'' operator supported on a narrow strip. 
In Sec. VII we explain why partial self-correction is impossible for two-dimensional local stabilizer codes  with distance $d=\Omega(L)$. 
In Sec. VIII we show that the logical string operator in a two-dimensional local commuting projector code can be chosen to be slightly entangling across any cut through the string.
Sec. IX contains our conclusions. 

\section{Background: stabilizer  and subsystem codes}
\label{sec:background}

A \emph{stabilizer code}  \cite{CalderbankRainsShorEtAl1997Quantum,Gottesman1996Class} embeds $k$ protected qubits in the Hilbert space of $n$ physical qubits. The code has a stabilizer group $S$, an abelian subgroup with $n-k$ independent generators of the $n$-qubit Pauli group $P$, and the code is the simultaneous eigenspace with eigenvalue 1 of all elements of $S$.

It is convenient to abelianize $P$ by ignoring the phase in the product of two Pauli operators, thus obtaining a $2n$-dimensional vector space over the binary field, which we also denote by $P$. The vector space $P$ is equipped with a symplectic form, such that two vectors are orthogonal if and only if the corresponding Pauli operators commute. If $G$ is a subgroup of $P$, we use the symbol $G$ to denote both the subgroup and the corresponding subspace of $P$.

Viewed as a vector space, $S$ is $(n-k)$-dimensional. We denote by $S^\perp$ the vector space orthogonal to $S$, which has dimension $2n - (n-k)$ = $n+k$. It can be decomposed as a direct sum of $S$ and a $2k$ dimensional vector space corresponding to the logical Pauli group, which acts nontrivially on the $k$ protected qubits. We define the weight of a Pauli operator as the number of qubits on which the operator acts nontrivially, and the distance $d$ of the stabilizer code is the minimum weight of a nontrivial logical operator (one contained in $S^\perp$ but not in $S$).

A \emph{subsystem code} \cite{Bacon2006Operator,Poulin2005Stabilizer} can be viewed as a stabilizer code with $k+g$ encoded qubits, but where only $k$ of these qubits are used to store protected quantum information. The stabilizer group $S$ together with Pauli operators acting on the $g$ unused qubits generate the code's \emph{gauge group} $G$. Equivalently, we may say that the subsystem code is defined by its gauge group $G\le P$, and that the code's stabilizer group $S=G\cap G^\perp$ is the subgroup of $G$ that commutes with all elements of $G$.

Logical operations in the subsystem code preserve the $2^k$-dimensional Hilbert space spanned by the $k$ protected qubits. We distinguish between \emph{bare} logical operators, which act trivially on the gauge qubits, and \emph{dressed} logical operators, which may act nontrivially on the gauge qubits as well as the protected qubits. Thus, nontrivial bare logical operators are in $G^\perp$ but not in $G$, while nontrivial dressed logical operators are in $S^\perp$ but not in $G$. The code distance $d$ is the minimum weight of a nontrivial dressed logical operator. 

A bare logical operator  $x\in G^\perp$ acts trivially on the protected qubits as well as the gauge qubits if and only if $x\in G^\perp \cap G= S$; hence we may regard $G^\perp/S$ as the group of bare logical operators. A dressed logical operator $x\in S^\perp$ acts trivially on the protected qubits (but perhaps nontrivially on the gauge qubits) if and only if  $x\in G$; hence we may regard $S^\perp / G$ as the group of dressed logical operators, where we regard two dressed logical operators as equivalent if they act the same way on the protected qubits. We denote by $[G]$ the dimension of the vector space $G$ (the number of independent generators of the corresponding group); by counting the number of independent bare logical operators, we find that the number $k$ of protected qubits satisfies
\begin{eqnarray*}
2k &=& [G^\perp/S]=[G^\perp] - [S] \\
 &=& [P] - [G] - [S]= 2n - [G] - [S].
\end{eqnarray*}
Similarly, by counting the number of independent dressed logical operators, we find
\begin{eqnarray*}
2k &=& [S^\perp/G]=[S^\perp] - [G] \\
 &=& [P] - [S] - [G]= 2n - [S] - [G].
\end{eqnarray*}
A stabilizer code is the special case of a subsystem code in which $G=S$, and in that case, $k = n - [S]$.

We will also consider stabilizer codes and subsystem codes of the CSS type \cite{CalderbankShor1996Good,Steane1996Multiple}, where each generator of the gauge group, and each logical operator, may be chosen to be either of the $X$-type or the $Z$-type. We use $P^X$ ($P^Z$) to denote the group of $X$-type ($Z$-type) Pauli operators, $G^X$ ($G^Z$) to denote the $X$-type ($Z$-type) gauge group, and $S^X$ ($S^Z$) to denote the $X$-type ($Z$-type) stabilizer group. We use ($G^X)^\perp$ to denote the subgroup of $P^Z$ that commutes with $G^X$, etc. Then the group of bare $Z$-type logical operators is $(G^X)^\perp/S^Z$ and the group of bare $X$-type logical operators is $(G^Z)^\perp/S^X$. Therefore the number $k$ of protected qubits is
\begin{eqnarray*}
k &=& [(G^X)^\perp/S^Z]=  n - [G^X] - [S^Z],\\
k &=& [(G^Z)^\perp/S^X]= n - [G^Z] - [S^X].
\end{eqnarray*}

We wish to study stabilizer codes in which the stabilizer generators are geometrically local and subsystem codes in which the gauge generators are geometrically local. To be concrete, we may imagine that the qubits reside at the vertices of a $D$-dimensional hypercubic lattice (with either open or periodic boundary conditions), and that each generator acts nontrivially only inside a hypercube (containing $w^D$ vertices) with linear size $w$. In fact our results can be easily extended to codes with geometrically local generators defined on any graph embedded in $D$-dimensional space. Note that for a subsystem code the stabilizer generators might be nonlocal even if the gauge generators are local. Some of our results also apply to a larger class of local codes that includes local stabilizer codes. For this class, which we call \emph{local commuting projector codes}, the code space is the simultaneous eigenspace with eigenvalue one of a set of mutually commuting geometrically local projection operators, where the projectors do not necessarily project onto eigenspaces of Pauli operators. A local stabilizer code, but not a local subsystem code, is a special case of a local commuting projector code. 

\section{Cleaning lemma for subsystem codes}

The Cleaning Lemma for subsystem codes relates the number of independent bare logical operators supported on a set of qubits $M$ to the number of independent dressed logical operators supported on the complementary set $M^c$. The concept of the Cleaning Lemma was introduced in \cite{BravyiTerhal2008no-go}, then generalized in \cite{YoshidaChuang2010Framework} and  \cite{Bravyi2010Subsystem}. Here we use ideas from \cite{YoshidaChuang2010Framework} to prove a version stated in \cite{Bravyi2010Subsystem}.
(See also \cite{WildeFattal2009Nonlocal}.)
As in the Sec. \ref{sec:background}, we will regard a subgroup of the Pauli group as a vector space, allowing us to obtain the Cleaning Lemma from straightforward dimension counting.

We use $P_A$ to denote the subgroup of the Pauli group $P$ supported on a set $A$ of qubits; likewise for any subgroup $G$ of the Pauli group $G_A= G \cap P_A$, is the subgroup of $G$ supported on $A$.
We denote by $\Pi_A : P \to P_A$ the restriction map that maps a Pauli operator to its restriction supported on the set $A$, and we use $|A|$ to denote the number of qubits contained in $A$; thus $[P_A] = 2|A|$.

If we divide $n$ qubits into two complementary sets $A$ and $B$, then a subgroup $G$ of $P$ can be decomposed into $G_A$, $G_B$, and a ``remainder,'' as follows:

\begin{lem}\emph{(Decomposition of Pauli subgroups)}
Suppose that $A$ and $B$ are complementary sets of qubits. Then for any subgroup $G$ of the Pauli group,
\begin{equation*}
 G = G_A \oplus G_B \oplus G'
\end{equation*}
for some $G'$, where
\begin{align*}
 [ (G^\perp)_A ] &= 2|A| - [G_A] - [G'] ,\\
 [ (G^\perp)_B ] &= 2|B| - [G_B] - [G']
\end{align*}
\end{lem}

\begin{proof}
If $V$ is a vector space and $W$ is a subspace of $V$, then there is a vector space $V'$ such that $V=W\oplus V'$; we may choose $V'$ to be the span of the basis vectors that extend a basis for $W$ to a basis for $V$. 
Since $G_A$ and $G_B$ are disjoint, i.e., $G_A \cap G_B = \{0\}$, $G_A\oplus G_B$ is a subspace of $G$, and thus there exists an auxiliary vector space $G' \leq G$ such that
\[
 G = G_A \oplus G_B \oplus G'.
\]
The choice of $G'$ is not canonical, but we need only its existence. Since the restriction map $\Pi_A$ obviously annihilates $G_B$, we may regard it as a map from $G_A\oplus G'$ onto $\Pi_A G$. In fact this map is injective. Note that if $\Pi_A x = 0$ for some $x \in G_A\oplus G'$. then since $P=P_A\oplus P_B$ it must be that $x \in G_B$. But because the sum is direct, i.e. $G_B \cap (G_A\oplus G') = \{0\}$, it follows that $x = 0$, which proves injectivity. Hence $\Pi_A: G_A\oplus G'\to \Pi_A G$ is an isomorphism. Now, we may calculate $(G^\perp)_A$ by solving a system of linear equations. Noting that $x \in P_A$ is contained in $G^\perp$ if and only if $x$ commutes with the restriction to $A$ of each element of $G$, we see that the number of independent linear constraints is $[\Pi_A G] = [G_A] + [G']$; hence $[(G^\perp)_A]=[P_A] - [G_A] - [G']= 2|A| - [G_A] - [G']$. Likewise, $\Pi_B: G_B\oplus G'\to \Pi_B G$ is also an isomorphism, and hence $[(G^\perp)_B]=[P_B] - [G_B] - [G']= 2|B| - [G_B] - [G']$.
\end{proof}

Now we are ready to state and prove the Cleaning Lemma. For a subsystem code, let $g_{\rm bare}(M)$ be the number of independent non-trivial bare logical operators supported on $M$, and let $g(M)$ be the number of independent non-trivial dressed logical operators supported on $M$, i.e.,
\begin{align*}
 g_{\rm bare}(M) &= [G^\perp \cap P_M / S_M ] = [(G^\perp)_M/S_M], \\
 g(M)        &= [S^\perp \cap P_M / G_M ]= [(S^\perp)_M/G_M].
\end{align*}
Likewise, for a CSS subsystem code, let $g_{\rm bare}^X(M)$ be the number of independent non-trivial bare $X$-type logical operators supported on $M$, and let $g^X(M)$ be the number of independent non-trivial dressed $X$-type logical operators supported on $M$, i.e.,
\begin{align*}
 g_{\rm bare}^X(M) &= [(G^Z)^\perp \cap P^X_M / S^X_M], \\
 g^X(M)        &= [(S^Z)^\perp \cap P^X_M / G^X_M],
\end{align*}
and similarly for the $Z$-type logical operators.

\begin{lem}\emph{(Cleaning Lemma for subsystem codes)}
For any subsystem code, we have
\begin{equation}
 g_{\rm bare}(M) + g(M^c) = 2k ,
\end{equation}
where $M$ is any set of qubits and $M^c$ is its complement. 
Moreover, for a CSS subsystem code
\begin{equation}
g_{\rm bare}^X(M) + g^Z(M^c) = k = g_{\rm bare}^Z(M) + g^X(M^c).
\end{equation}
\label{lem:counting_op}
\end{lem}
\begin{proof}
We use Lemma 1 to prove the Cleaning Lemma by a direct calculation:
\begin{align*}
g_{\rm bare}(M)
&= [(G^\perp)_M  / S_M] \\
&= 2|M| - [G_M] - [G'] - [S_M] ,
\end{align*}
and
\begin{align*}
g(M^c) 
&= [(S^\perp)_{M^c} / G_{M^c}] \\
&= 2|M^c| - [S_{M^c}] - [S'] - [G_{M^c}] .
\end{align*}
Summing, we find
\begin{align*}
g_{\rm bare}(M) + g(M_c) 
&= 2|M| + 2|M_c| \\
&-([G_M] + [G_{M_c}] + [G'])\\
&-([S_M] + [S_{M_c}] + [S'])
\end{align*}
and invoking Lemma 1 once again,
\begin{align*}
g_{\rm bare}(M) + g(M_c) 
&= 2n -[G] - [S] = 2k ,
\end{align*}
which proves the claim for general subsystem codes.
For the CSS case, we apply the analogue of Lemma 1 to the $X$-type and $Z$-type Pauli operators, finding
\begin{align*}
g^Z_{\rm bare}(M)
&= [ (G^X)^\perp \cap P^Z_{M} / S^Z_{M} ] \\
&= |M| - [G^X_{M}] - [(G^X)'] - [S^Z_{M}]
\end{align*}
and also 
\begin{align*}
g^X(M^c)
&= [ (S^Z)^\perp \cap P^X_{M^c} / G^X_{M^c} ] \\
&= |M^c| - [S^Z_{M^c}] - [(S^Z)'] - [G^X_{M^c}]. 
\end{align*}
Summing and using Lemma 1 we have
\begin{align*}
g_{\rm bare}^Z(M) + g^X(M^c) 
&= n - [G^X]-[S^Z] =k ;
\end{align*}
a similar calculation yields 
\begin{align*}
g_{\rm bare}^X(M) + g^Z(M^c) 
&= n - [G^Z]-[S^X] =k ,
\end{align*}
proving the claim for CSS subsystem codes.
\end{proof}

Of course, for a stabilizer code there is no distinction between bare and dressed logical operators; the statement of the Cleaning Lemma becomes
\[
g(M) + g(M^c) = 2k
\]
for general stabilizer codes, and
\[
g^X(M) + g^Z(M^c) = k
\]
for CSS stabilizer codes.

To understand how the Cleaning Lemma gets its name, note that it implies that if no bare logical operator can be supported on the set $M$ then all dressed logical operators can be supported on its complement $M^c$. That is, any of the code's dressed logical Pauli operators can be ``cleaned up'' by applying elements of the gauge group $G$. The cleaned operator acts the same way on the protected qubits as the original operator (though it might act differently on the gauge qubits), and acts trivially on $M$. 

We say that a region $M$ is \emph{correctable} if erasure of the qubits in $M$ is a correctable error. For a subsystem code, it follows that no nontrivial dressed logical operators are supported on $M$ if $M$ is correctable; hence $g(M)=0$ and thus $g_{\rm bare}(M)=0$. The Cleaning Lemma then asserts that all dressed logical operators can be supported on $M^c$. Let us say that two dressed logical operators $x$ and $y$ are \emph{equivalent} if $x=yz$ and $z$ is an element of the gauge group $G$, so that $x$ and $y$ act the same way on the protected qubits. We have obtained:
\begin{lem}\emph{(Cleaning Lemma for dressed logical operators)}
\label{lem:clean-region}
For any subsystem code, if $M$ is a correctable region and $x$ is a dressed logical operator, then there is a dressed logical operator $y$ supported on $M^c$ that is equivalent to $x$.
\end{lem}

\section{Operator tradeoff for local subsystem codes}
\label{sec:subsystem-tradeoff}
In this section we consider local subsystem codes with qubits residing at the sites of a $D$-dimensional hypercubic lattice $\Lambda$. The code has \emph{interaction range} $w$, meaning that the generators of the gauge group $G$ can be chosen so that each generator has support on a hypercube containing $w^D$ sites.

\begin{defn}
\label{defn:boundary}
Given a set of gauge generators for a subsystem code, and a set of qubits $M$, let $M'$ denote the support of all the gauge generators that act nontrivially on $M$. The \emph{external boundary} of $M$ is $\partial_+ M = M' \cap M^c$, where $M^c$ is the complement of $M$, and the \emph{internal boundary} of $M$ is  $\partial_- M = \left(M^c\right)' \cap M$.  The \emph{boundary} of $M$ is $\partial M=\partial_+M\cup\partial_-M$, and the \emph{interior} of $M$ is $M^\circ = M \setminus \partial_- M$.
\end{defn}

Recall that a region (\emph{i.e.}, set of qubits) $M$ is said to be \emph{correctable} if no nontrivial dressed logical operation is supported on $M$, in which case erasure of $M$ can be corrected. Since the code distance $d$ is defined as the minimum weight of a dressed logical operator, $M$ is certainly correctable if $|M| < d$. But in fact much larger regions are also correctable, as follows from this lemma: 

\begin{lem}\emph{(Expansion Lemma for local subsystem codes)}
For a local subsystem code, if $M$ and $A$ are both correctable, where $A$ contains $\partial M$, then $M\cup A$ is correctable.
\label{lem:subsystem-extend}
\end{lem}
\begin{proof}
Given a subsystem code $\mathcal{C}$ with gauge group $G$, we may define a subsystem code $\mathcal{C}_{M^c}$ on $M^c$ with gauge group $\Pi_{M^c}G$, where $\Pi_{M^c}$ maps a Pauli operator to its restriction supported on $M^c$. We note that a Pauli operator $x$ supported on $M^c$ is a bare logical operator for $\mathcal{C}$ if and only if $x$ is a bare logical operator for $\mathcal{C}_{M^c}$; that is, $x$ commutes with all elements of $G$ if and only if it commutes with all elements of the restriction of $G$ to $M^c$. 

Furthermore, if $x$ is a dressed logical operator for $\mathcal{C}_{M^c}$ supported on $\partial_+M$, then $x$ can be extended to a dressed logical operator $\bar x$ for $\mathcal{C}$ supported on $\partial M$. Indeed, suppose $x=yz$, where $y$ is a bare logical operator for $\mathcal{C}_{M^c}$ (and hence also a bare logical operator for $\mathcal{C}$ supported on $M^c$), while $z$ is an element of the gauge group $\Pi_{M^c} G$ of $\mathcal{C}_{M^c}$. Then $z$ can be written as a product $z=\prod_i g_i$ of generators of $\Pi_{M^c} G$, each of which can be expressed as $g_i = \Pi_{M^c} \bar g_i$, where $\bar g_i$ is a generator of $G$ supported on $M^c\cup  \partial_-M$. Thus $\bar x = y\prod_i\bar g_i$ is a dressed logical operator for $\mathcal{C}$ supported on $\partial M$. 

It follows that if $\partial M$ is correctable for the code $\mathcal{C}$ (\emph{i.e}, code $\mathcal{C}$ has no nontrivial dressed logical operators supported on $\partial M$), then $\partial_+ M$ is correctable for the code $\mathcal{C}_{M^c}$ ($\mathcal{C}_{M^c}$ has no nontrivial dressed logical operators supported on $\partial_+ M$). By similar logic, if $A$ is correctable for $\mathcal{C}$ and contains $\partial M$, then $A\cap M^c$ is correctable for $\mathcal{C}_{M^c}$.

Suppose now that the code $\mathcal{C}$ has $k$ encoded qubits and that $M$ is correctable, \emph{i.e.} $g^{(\mathcal{C})}(M)=0$. Therefore, applying Lemma \ref{lem:counting_op} to the code $\mathcal{C}$, $g_{\rm bare}^{(\mathcal{C})}(M^c)= 2k$. Suppose further that the set $A$ containing $\partial M$ is correctable for $\mathcal{C}$, implying that $A\cap M^c$ is correctable for $\mathcal{C}_{M^c}$, \emph{i.e.} $g^{(\mathcal{C}_{M^c})}(A\cap M^c)=0$. Then applying Lemma \ref{lem:counting_op} to the code $\mathcal{C}_{M^c}$, we conclude that $g_{\rm bare}^{(\mathcal{C}_{M^c})}(M^c\setminus A)=2k$. Since each bare logical operator for $\mathcal{C}_{M^c}$, supported on $M^c\setminus A$, is also a bare logical operator for $\mathcal{C}$, supported on $M^c\setminus A$, we can now apply Lemma \ref{lem:counting_op} once again to the code $\mathcal{C}$, using the partition into $M^c\setminus A$ and $M\cup A$, finding $g^{(\mathcal{C})}(M \cup A)=0$. Thus $M \cup A$ is correctable. 
\end{proof}

If the interaction range is $w$, and $M$ is a correctable hypercube with linear size $l-2(w-1)$, then we may choose $A\supseteq \partial M$ so that $M\cup A$ is a hypercube with linear size $l$ and $M\setminus A$ is a hypercube with linear size $l - 4(w-1)$. Then $A$ contains
\[
|A| = l^D - \left[l-4(w-1)\right]^D \le 4(w-1)Dl^{D-1}
\]
qubits, and $A$ is surely correctable provided $|A|<d$, where $d$ is the code distance. Suppose that $d>1$, so a single site is correctable. Applying Lemma \ref{lem:subsystem-extend} repeatedly, we can build up larger and larger correctable hypercubes, with linear size $1 + 2(w-1), 1+ 4(w-1), 1+ 6(w-1), \dots$. 
This process continues as long as $|A|< d$. We conclude:

\begin{lem}\emph{(Holographic Principle for local subsystem codes)}
\label{lem:subsystem-hypercube}
For a $D$-dimensional local subsystem code with interaction range $w>1$ and distance $d>1 $, a hypercube with linear size $l$ is correctable if 
\begin{equation}\label{eq:hypercube-size}
4(w-1)Dl^{D-1} < d.
\end{equation}
\end{lem}

\noindent Thus (roughly speaking) for the hypercube to be correctable it suffices for its $\left[2(w-1)\right]$-thickened \emph{boundary}, rather than its volume, to be smaller than the code distance. Bravyi \cite{Bravyi2010Subsystem} calls this property ``the holographic principle for error correction,'' because the \emph{absence} of information encoded at the boundary of a region ensures that no information is encoded in the ``bulk.'' For local stabilizer codes, the criterion for correctability is slightly weaker than for local subsystem codes, as we discuss in Appendix \ref{app:holographic_lemma_stabilizer_codes}.

Now we are ready to prove our first tradeoff theorem.

\begin{figure}
\centering
\includegraphics[width=0.42\textwidth]{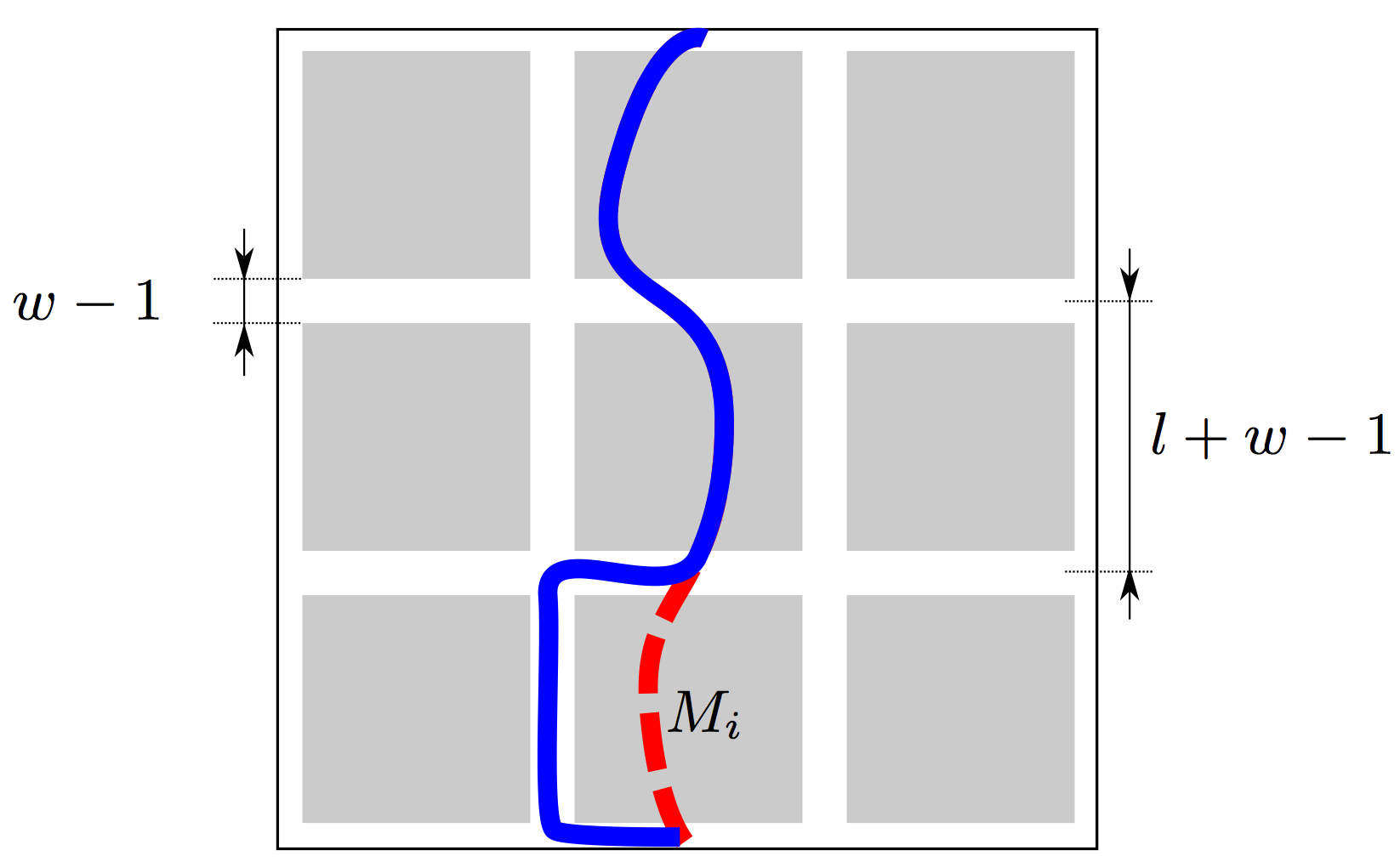}
\caption{(Color online) Lattice covering used in the proof of Theorem 1, shown in two dimensions. Each gray square is $l\times l$ and the white gap between squares has width $w-1$. The solid blue curve represents the support of a nontrivial logical operator; because the square $M_i$ is correctable, this square can be ``cleaned'' --- we can find an equivalent logical operator supported on $M_i^c$, the complement of $M_i$. When all squares are cleaned, the logical operator is supported on the narrow strips between the squares. }
\label{fig:cleaning}
\end{figure}

\begin{theorem}\emph{(Tradeoff Theorem for local subsystem codes)}
For a local subsystem code in $D\ge 2$ dimensions with interaction range $w>1$ and distance $d\gg w$, defined on a hypercubic lattice with linear size $L$, every dressed logical operator is equivalent to an operator with weight $\tilde d$ satisfying
\begin{equation}\label{eq:tradeoff-bound}
  \tilde d {d}^{1/(D-1)} < c L^D,
\end{equation}
where $c$ is a constant depending on $w$ and $D$.
\label{thm:subsystem-tradeoff}
\end{theorem}

\begin{proof}
As shown in Fig. \ref{fig:cleaning}, we fill the lattice with hypercubes, separated by distance $w-1$, such that each hypercube has linear size $l$ satisfying eq.\eqref{eq:hypercube-size}. (By ``distance'' we mean the number of sites in between --- \emph{e.g.} we say that adjacent sites are ``distance zero'' apart.) Thus no gauge generator acts nontrivially on more than one hypercube, and each hypercube is correctable by Lemma \ref{lem:subsystem-hypercube}. Consider any nontrivial dressed logical operator $x$, and label the hypercubes $\{M_1, M_2, M_3, \dots\}$. By Lemma 3 there exists a gauge operator $y_i$ that ``cleans'' the logical operator in the hypercube $M_i$, \emph{i.e.}, such that $xy_i$ acts trivially in $M_i$. Furthermore, since no gauge generator acts nontrivially on more than one hypercube, we can choose $y_i$ so that it acts trivially in all other hypercubes. Taking the product of all the $y_i$'s we construct  a gauge operator that cleans all hypercubes simultaneously; thus $\tilde x= x\prod_i y_i$ is equivalent to $x$ and supported on the complement of the union of hypercubes $M=\cup_i M_i$. Therefore, the weight $\tilde d$ of $\tilde x$ is upper bounded by $|M^c|$. 

The lattice is covered by hypercubes of linear size $l+(w-1)$, each centered about one of the $M_i$'s. There are $L^D/\left[l+(w-1)\right]^D$ such hypercubes in this union, each containing no more than $\left[l+(w-1)\right]^D - l^D \le (w-1)D\left[l+(w-1)\right]^{D-1}$ elements of $M^c$. Thus 
\begin{align*}
\tilde d \le |M^c| &\le (w-1)D\left[l+(w-1)\right]^{D-1}\frac{L^D}{\left[l+(w-1)\right]^D} \nonumber\\
&= \frac{(w-1)D}{l+(w-1)}L^D.
\end{align*}
We optimize this upper bound on $\tilde d$ by choosing $l$ to be the largest integer such that a hypercube with linear size $l$ is known to be correctable,  \emph{i.e.}, satisfying
\[
l < \left(\frac{d}{4(w-1)D}\right)^{1/(D-1)},
\]
thus obtaining eq.\eqref{eq:tradeoff-bound}. Note that eq.\eqref{eq:tradeoff-bound} is trivial if $d$ is a constant independent of $L$, since the weight $\tilde d$ cannot be larger than $L^D$.
\end{proof}

\section{Operator tradeoff for local commuting projector codes}

In this section we consider a local commuting projector code, defined as the simultaneous eigenspace with eigenvalue one of a set of commuting projectors. As in Sec. \ref{sec:subsystem-tradeoff} we assume that the qubits reside on a hypercubic lattice $\Lambda$ and that each projector acts trivially outside a hypercube of linear size $w$, where $w$ is the interaction range. By a \emph{logical operator} we mean any transformation that preserves the code space, and we say that two logical operators are \emph{equivalent} if they have the same action on the code space. The weight of a logical operator is the number of qubits on which it acts nontrivially. We say that a set of qubits $M$ is correctable if erasure of $M$ can be reversed by a trace-preserving completely positive recovery map. The distance $d$ of the code is the minimum size of a noncorrectable set of qubits. 

Bravyi, Poulin, and Terhal \cite{BravyiPoulinTerhal2010Tradeoffs} proved some useful properties of these codes. To state their results, we use the definition

\begin{defn}
\label{defn:boundary-projector}
Given a set of commuting projectors defining a code, and a set of qubits $M$, let $M'$ denote the support of all the projectors that act nontrivially on $M$. The \emph{external boundary} of $M$ is $\partial_+ M = M' \cap M^c$, where $M^c$ is the complement of $M$, and the \emph{internal boundary} of $M$ is  $\partial_- M = \left(M^c\right)' \cap M$.  The \emph{boundary} of $M$ is $\partial M=\partial_+M\cup\partial_-M$, and the \emph{interior} of $M$ is $M^\circ = M \setminus \partial_- M$.
\end{defn}

\begin{lem}
\label{lem:disentangling}
\emph{(Disentangling Lemma \cite{BravyiPoulinTerhal2010Tradeoffs})}
Consider a local commuting projector code and suppose that $M$ and $\partial_+M$ are both correctable regions. Then there exists a unitary operator $U_{\partial M}$ acting only on the boundary $\partial M$ such that, for any pure code vector $|\psi\rangle$,
\begin{equation}\label{eq:disentangling}
 U_{\partial M} \ket{\psi} 
= \ket{\phi_M} \otimes \ket{\psi'_{M^c}} .
\end{equation}
Here $\ket{\phi_M}$, supported on $M$, does not depend on the code vector $|\psi\rangle$, while $\ket{\psi'_{M^c}}$, supported on $M^c$, does depend on $|\psi\rangle$.
\end{lem}

\noindent
The Disentangling Lemma says that, if $M$ and $\partial_+M$ are both correctable, then the entanglement of code vectors across the cut between $M$ and $M^c$ is localized in $\partial M$ and can be removed by a unitary transformation acting on only $\partial M$. Furthermore, in the resulting product state, no information distinguishing one code vector from another is available in $M$. This Lemma has a simple but important corollary: 

\begin{lem}
\emph{(Expansion Lemma for local commuting projector codes \cite{BravyiPoulinTerhal2010Tradeoffs})}
For a local commuting projector code, if $M$ and $A$ are both correctable, where $A$ contains $\partial M$, then $M\cup A$ is correctable.
\label{lem:commuting-projector-extend}
\end{lem}
\begin{proof}
By eq.(\ref{eq:disentangling}), if $A$ is erased the resulting state on $M\setminus A$ is independent of the code vector $|\psi\rangle$; all the information needed to reconstruct $|\psi\rangle$ resides in $M^c\setminus A$. Therefore, we can erase $M\setminus A$ as well without compromising our ability to reconstruct $|\psi\rangle$; that is, $M\cup A$ is correctable.
\end{proof}

\noindent Definition \ref{defn:boundary-projector} and Lemma \ref{lem:commuting-projector-extend} for commuting projector codes are parallel to Definition \ref{defn:boundary} and Lemma \ref{lem:subsystem-extend} for subsystem codes. Arguing as in the proof of Lemma \ref{lem:subsystem-hypercube}, we see that one consequence is a holographic principle for these codes:
\begin{lem}\emph{(Holographic Principle for local commuting projector codes)}
\label{lem:projector-hypercube}
For a $D$-dimensional local commuting projector code with interaction range $w>1$ and distance $d>1 $, a hypercube with linear size $l$ is correctable if 
\begin{equation}\label{eq:hypercube-size-projector}
4(w-1)Dl^{D-1} < d.
\end{equation}
\end{lem}

We will need an analog of the Cleaning Lemma to analyze the logical operator tradeoff for local commuting projector codes; it can be derived from the Disentangling Lemma. 

\begin{lem}
\emph{(Cleaning Lemma for local commuting projector codes)}
Consider a local commuting projector code, and suppose that $M$ and $\partial_+ M$ are both correctable. For any logical operator $W$ there exists an equivalent logical operator $V$ supported on the complement of the interior $M^{\circ}$ of $M$. If $W$ is an isometry, then $V$ can be chosen to be unitary.
\label{lem:cleaning-projector}
\end{lem}
\begin{proof}
Let us name the regions:
\begin{align*}
 A =& M^\circ = M \setminus \partial_- M, & B =& \partial_- M, \\
 C =& \partial_+ M,                     & D =& (ABC)^c.
\end{align*}
Let $\{ \ket{\alpha_i} \}$ be an orthonormal basis for the code space. By Lemma \ref{lem:disentangling}, there exists a unitary transformation $U_{BC}$, and vectors $\ket{\phi}_{AB}, \{\ket{\alpha'_i}_{CD}\}$ such that
\[
 \ket{\alpha_i} = U_{BC} \ket{\phi}_{AB} \otimes \ket{\alpha'_i}_{CD} ,
\]
where the normalized vector $|\phi\rangle_{AB}$ does not depend on $i$ and the vectors $\{|\alpha'_i\rangle_{CD}\}$ are normalized and mutually orthogonal. Because $W$ is a logical operator, $\ket{\beta_i} \equiv W \ket{\alpha_i}$ is also a code vector, and therefore
\[
 \ket{\beta_i} = U_{BC} \ket{\phi}_{AB} \otimes \ket{\beta'_i}_{CD}
\]
where $\{\ket{\beta'_i}_{CD}\}$ is another set of vectors; if $W$ is an isometry then these vectors, too, are normalized and mutually orthogonal. Define a transformation $V'$ by $\ket{\beta'_i}_{CD} = V' \ket{\alpha'_i}_{CD}$, and choose an arbitrary extension so that $V'$ becomes an operator on $CD$. If $W$ is an isometry, then this extension $V'_{CD}$ can be chosen to be unitary. We now have 
\begin{align*}
 W \ket{\alpha_i}= |\beta_i\rangle &= U_{BC} (I_{AB} \otimes V'_{CD}) \ket{\phi}_{AB}\otimes \ket{\alpha'_i}_{CD}\\
&=U_{BC} (I_{AB} \otimes V'_{CD}) U_{BC}^\dagger \ket{\alpha_i}
\end{align*}
for all $i$. Defining 
\[
V_{{(M^\circ})^c} = U_{BC} (I_{AB} \otimes V'_{CD}) U_{BC}^\dagger,
\]  
we observe that $V_{{(M^\circ})^c}$ acts trivially on $A=M^\circ$ and has the same action on code vectors as $W$, completing the proof.
\end{proof}

To prove the tradeoff theorem we will need a further lemma establishing that a union of correctable sets is correctable under suitable conditions. Recall that we say a set of qubits $M$ is correctable if and only if erasure of $M$ can be corrected. Equivalently, $M$ is correctable if and only if, for any operator $\mathcal{O}$ supported on $M$,
\begin{equation}\label{eq:correctable}
\Pi \mathcal{O} \Pi = c_{\mathcal{O}} \Pi
\end{equation}
where $\Pi$ denotes the projector onto the code space and $c_{\mathcal{O}}$ is a constant (possibly zero) depending on $\mathcal{O}$ \cite{Gottesman2009Introduction}. 

\begin{lem}
\emph{(The union of separated correctable regions is correctable)}
For a local commuting projector code, suppose that  $M$ and $N$ are correctable regions such that no projector acts nontrivially on both $M$ and $N$. Then $M \cup N$ is also correctable.
\label{lem:extension_of_correctable_region}
\end{lem}
\noindent A weaker version of this lemma was proved in \cite{BravyiPoulinTerhal2010Tradeoffs}.
\begin{proof}
Let $\mathcal{S}$ be the set of local commuting projectors that define the code. We denote by $\mathcal{S}'_N$ the set of projectors in $\mathcal{S}$ that act nontrivially on $N$. Define
\begin{align*}
\Pi_{N'} &= \prod_{\Pi_a \in \mathcal{S}'_N} \Pi_a ,\\
\Pi_{N'}^c    &= \prod_{\Pi_a \in \mathcal{S} \setminus \mathcal{S}'_N} \Pi_a ,
\end{align*}
and note that the projector onto the code space is 
\[\Pi = \prod_{\Pi_a\in\mathcal{S}}\Pi_a=\Pi_{N'} \Pi_{N'}^c.
\]
Also note that the support of $\Pi_{N'}$ does not intersect $M$ and the support of $\Pi_{N'}^c$ does not intersect $N$. Let $\mathcal{O}$ be an arbitrary operator supported on $M \cup N$; we will show that $\mathcal{O}$ satisfies eq.\eqref{eq:correctable}. Since $M$ and $N$ are disjoint, $\mathcal{O}$ has a Schmidt decomposition
\[
\mathcal{O} = \sum_\alpha \mathcal{O}_M^\alpha \otimes \mathcal{O}_N^\alpha
\]
where each $\mathcal{O}_M^\alpha$ is supported on $M$ and each $\mathcal{O}_N^\alpha$ is supported on $N$. Since $\Pi_{N'}$ commutes with $\mathcal{O}_M^\alpha$ and $\Pi_{N'}^c$ commutes with $\mathcal{O}_N^\alpha$,
\begin{align*}
 \Pi \mathcal{O} \Pi 
&= \sum_\alpha \Pi (\Pi_{N'}) \mathcal{O}_M^\alpha \mathcal{O}_N^\alpha (\Pi_{N'}^c) \Pi \\
&= \sum_\alpha \Pi \mathcal{O}_M^\alpha (\Pi_{N'}) (\Pi_{N'}^c) \mathcal{O}_N^\alpha \Pi \\
&= \sum_\alpha \left(\Pi \mathcal{O}_M^\alpha \Pi \right)\left(\Pi\mathcal{O}_N^\alpha \Pi\right) \\
&= \sum_\alpha c_{ \mathcal{O}_M^\alpha } c_{ \mathcal{O}_N^\alpha } \Pi \\
&= c_{\mathcal{O}} \Pi
\end{align*}
where in the fourth equality we used the correctability of $M$ and $N$. Thus $\mathcal{O}$ obeys eq.\eqref{eq:correctable}, and $M \cup N$ is correctable.
\end{proof}

Now we are ready to state and prove our second tradeoff theorem.

\begin{theorem}\emph{(Tradeoff Theorem for local commuting projector codes)}
For a local commuting projector code in $D\ge 2$ dimensions with interaction range $w>1$ and distance $d\gg w$, defined on a hypercubic lattice with linear size $L$, every logical operator is equivalent to an operator with weight $\tilde d$ satisfying
\begin{equation}\label{eq:tradeoff-bound-again}
  \tilde d {d}^{1/(D-1)} < c L^D,
\end{equation}
where $c$ is a constant depending on $w$ and $D$.
\label{thm:commuting_tradeoff}
\end{theorem}
\begin{proof}
The proof is similar to the proof of  Theorem \ref{thm:subsystem-tradeoff}. We fill the lattice with hypercubes, separated by distance $w-1$, where each hypercube $M_i$ has linear size $l$ sufficiently small so that $M_i$ and $\partial_+M_i$ are both correctable. Applying Lemma \ref{lem:extension_of_correctable_region} repeatedly, we conclude that the union $M$ of all $M_i$ is correctable, and the union $\partial_+M$ of all $\partial_+ M_i$ is correctable. 

For any logical operator, Lemma \ref{lem:cleaning-projector} now ensures the existence of an equivalent logical operator supported outside the interior $M^\circ$ of $M$, and hence the weight $\tilde d$ of this equivalent logical operator is bounded above by $|(M^\circ)^c|$. The lattice is covered by hypercubes with linear size $l + (w-1)$, each centered about one of the $M_i$, and there are $L^D/\left[l+(w-1)\right]^D$ such hypercubes, each containing no more than 
\begin{align*}
\left[l+(w-1)\right]^D - \left[l-2(w-1)\right]^D \nonumber\\
\le 3(w-1)D\left[l+(w-1)\right]^{D-1}
\end{align*}
elements of $(M^\circ)^c$; therefore,
\begin{align*}
\tilde d &\le |(M^\circ)^c| \\
&\le 3(w-1)D\left[l+(w-1)\right]^{D-1}\frac{L^D}{\left[l+(w-1)\right]^D}\\
& = \frac{3(w-1)D}{l+(w-1)}L^D.
\end{align*}

To ensure that $M_i$ and $\partial_+ M_i$ are correctable, it suffices that $|\partial M_i| < d$, where $d$ is the code distance, or
\begin{align*}
|\partial M_i| 
&\le \left[l+2(w-1)\right]^D - \left[l-2(w-1)\right]^D \\
&\le 4(w-1)D\left[l+2(w-1)\right]^{D-1} < d.
\end{align*}
We choose the largest such integer value of $l$, obtaining eq.\eqref{eq:tradeoff-bound-again}.
\end{proof}

\section{``String'' operators for local commuting projector codes}
\label{sec:string}
Because the code distance $d$ is defined as the size of the smallest noncorrectable set, and because a set supporting a nontrivial logical operator is noncorrectable, we have $d\le \tilde d$ and hence Theorem \ref{thm:commuting_tradeoff} implies 
\[
d = O(L^{D-1}).
\] 
In fact we can make a stronger statement, specifying the geometry of a region that supports a nontrivial logical operator with weight $O(L^{D-1})$. On the  hypercube $\{1,2,3,\dots L\}^D$, we refer to the set $\{i, i+1, \dots, i+r-1\}\times \{1,2,3,\dots,L\}^{D-1}$ as a \emph{slab} of width $r$. Let us say that a code is nontrivial if the code space dimension is greater than one. Then:

\begin{lem}
\emph{(Existence of a noncorrectable thin slab)}
For a nontrivial local commuting projector code in $D\ge 1$ dimensions with interaction range $w>1$, there is a noncorrectable slab of width $3(w-1)$.
\label{lem:slab}
\end{lem}
\begin{proof}
Suppose, contrary to the claim, that any slab of width $3(w-1)$ is correctable. Choose a correctable slab $M$ of width $3(w-1)$. The boundary $\partial M$ of $M$ is contained in two slabs $M_L$ and $M_R$, each of width $2(w-1)$. Hence $M_L$ and $M_R$ are both correctable, and since $M$ has width $3(w-1)$, $M_L$ and $M_R$ are separated by $w-1$. Therefore, no local projector acts on both $M_L$ and $M_R$, and by Lemma \ref{lem:extension_of_correctable_region}, $M_L \cup M_R \supseteq \partial M$ is correctable. Then Lemma \ref{lem:commuting-projector-extend} implies that the slab $M \cup M_L \cup M_R$ of width $5(w-1)$ is correctable. Repeating the argument, we see that if a slab $M$ of width $r$ is correctable, so is the slab of width $r+2(w-1)$ containing $M$.

If the system obeys open boundary conditions, then by induction the entire lattice is correctable. If the lattice is periodic, we may consider two thick correctable slabs $M_1, M_2$ such that $M_1 \cup M_2$ is the entire lattice and $\partial M_1 \subseteq M_2$; in that case Lemma \ref{lem:commuting-projector-extend} implies that the entire lattice $M_1 \cup M_2$ is correctable. For either type of boundary condition, then, there are no nontrivial logical operators at all. But we assumed that the code is nontrivial, and therefore reach a contradiction. 
\end{proof}

It follows from Lemma \ref{lem:slab} that the distance $d$ of a local commuting projector code satisfies
\[
d \le  3(w-1) L^{D-1}.
\] 
It was previously known that $d \le w L^{D-1}$ for a local stabilizer code  \cite{BravyiTerhal2008no-go,KayColbeck2008Quantum} and $d \le 3w L^{D-1}$ for a local subsystem code  \cite{BravyiTerhal2008no-go}. 

Now we may wonder about the geometry of a set that supports a nontrivial logical operator. For a subsystem code, there is a nontrivial logical operator supported by any noncorrectable set, but this statement is not true for general codes (see Appendix \ref{app:counter_example_nolop_noncorrectable}).
We say that an operator $\mathcal{O}$ is a \emph{logical} operator if it preserves the code space, and that it is a \emph{nontrivial} logical operator if it preserves the code space and its restriction to the code space is not proportional to the identity. 
From the definition of correctability, then, $M$ is not correctable if it supports a nontrivial logical operator. But for some codes the converse is false. If $M$ is not correctable, then an operator $\mathcal{O}$ exists that fails to satisfy eq.\eqref{eq:correctable}; however $\mathcal{O}$ might not preserve the code space.

But for a local commuting projector code, a correctable set can be extended to a slightly larger set that does support a nontrivial logical operator. Suppose the code is the simultaneous eigenspace with eigenvalue one of a set of commuting projectors $\mathcal{S}=\{\Pi_a\}$. For any set of qubits $M$, we define $M'$ as the support of all the projectors that act nontrivially on $M$. Then if $M$ is noncorrectable a nontrivial unitary logical operator is supported on $M'$.
\begin{lem}
\label{lem:projector-support}
\emph{(Support for nontrivial logical operator)}
For a commuting projector code, if the set $M$ is not correctable, then there is a nontrivial unitary logical operator supported on $M'$ that commutes with every projector in $\mathcal{S}$.
\end{lem}
\begin{proof}
Let $\Pi = \prod_{\Pi_a\in \mathcal{S}} \Pi_a$ be the projector onto the code space. We claim that there exists a Pauli operator $P_M$ supported on $M$ such that $\Pi P \Pi$ is not proportional to $\Pi$. Indeed, if $M$ is not correctable, then there exists an operator $\mathcal{O}_M$ supported on $M$ such that $\Pi \mathcal{O}_M \Pi \not\propto \Pi$. Expanding $\mathcal{O}_M = \sum_i c_i P^{(i)}_M $ as a linear combination of Pauli operators, we see that at least one Pauli operator $P^{(j)}_M$ must satisfy $\Pi P^{(j)}_M \Pi \not\propto \Pi$.

We denote by $\mathcal{S}'_M$ the set of projectors in $\mathcal{S}$ that act nontrivially on $M$, and 
define
\[
\Pi_{M'}= \prod_{\Pi_a\in \mathcal{S}'_M} \Pi_a.
\]
We claim that 
\[
\mathcal H = \Pi_{M'} P_M \Pi_{M'},
\]
is a nontrivial Hermitian logical operator supported on $M'$.

To see that $\mathcal{H}$ is a logical operator, note that if $\Pi_a \in \mathcal{S}'_M$,
then $\Pi_a \Pi_{M'} = \Pi_{M'} = \Pi_{M'}\Pi_a$, because $\Pi_a^2 = \Pi_a$; hence
\[
\Pi_a \mathcal H = \mathcal H = \mathcal H \Pi_a,
\]
\emph{i.e.}, $\Pi_a$ commutes with $\mathcal{H}$.
If $\Pi_a\not\in \mathcal{S}'_M$, then $\Pi_a$ is supported in the complement $M^c$ of $M$;
hence it commutes trivially with $P_M$, and therefore also with $\mathcal H$.
Since $\mathcal H$ commutes with each projector in $\mathcal{S}$,
it certainly commutes with $\Pi$ and hence preserves the code space. Furthermore, because
\[
\Pi \mathcal H \Pi = \Pi P_M \Pi,
\]
$\mathcal{H}$ acts on the code space in the same way as $\Pi P_M \Pi$, and therefore must be nontrivial.

Thus $U=\exp\left(-i\lambda \mathcal{H}\right)$ preserves the code space and is unitary for any real $\lambda$. Since $\mathcal{H}$, restricted to the code space, has at least two distinct eigenvalues, the same is true of $U$ for a generic choice of $\lambda$; \emph{i.e.}, $U$ is a nontrivial unitary logical operator.
\end{proof}
Lemmas \ref{lem:slab} and \ref{lem:projector-support} now imply:
\begin{theorem}
\emph{(A logical operator is supported on one thin slab)}
For a nontrivial local commuting projector code in $D\ge 1$ dimensions, with interaction range $w> 1$, there is a nontrivial unitary logical operator (commuting with all projectors) supported on a slab of width $5(w-1)$.
\label{thm:slab}
\end{theorem}
\noindent Note that, though the proof of Theorem~\ref{thm:slab} establishes the existence of a logical operator supported on a slab of constant width, it provides no algorithm for constructing the operator.

In $D=2$ dimensions, the slab becomes a strip of constant width stretching across the $L\times L$ code block, and the logical operator supported on the strip may be called a ``string'' operator. It was previously known that for $D=2$ a string operator can be supported on a strip of width $w$ in a local stabilizer code \cite{BravyiTerhal2008no-go,KayColbeck2008Quantum}, and of width $3w$ in a local subsystem code \cite{BravyiTerhal2008no-go}.

\section{Two-dimensional local stabilizer codes are not partially self correcting}
Theorems 1 and 2 constrain the weight of logical operators, but the proofs tell us more --- they specify the \emph{geometry} of a region that supports a logical operator. This geometry has further implications for the physical robustness of quantum memories.

Consider a subsystem code whose stabilizer group $S$ has a set of geometrically local generators $\{S_a\}$, where the qubits reside at the sites of a $D$-dimensional hypercubic lattice with linear size $L$. The generating set $\{S_a\}$ might be overcomplete, but we assume that the number of generators acting nontrivially on each qubit is a constant independent of $L$. The local Hamiltonian
\begin{equation}\label{eq:Hamiltonian}
H= -\sum_a \frac{1}{2}\left( S_a - I\right),
\end{equation}
has a $2^{k+g}$-fold degenerate ground state with energy $E=0$, where $k$ is the number of protected qubits and $g$ is the number of gauge qubits of the subsystem code --- each ground state is a simultaneous eigenstate with eigenvalue one of all elements of $\{S_a\}$. If a quantum memory governed by this Hamiltonian is subjected to thermal noise, how well protected is the $2^k$-dimensional code space? 

If $|\psi\rangle$ is a zero-energy eigenstate of $H$ and $x\in P$, then $x|\psi\rangle$ is an eigenstate of $H$ with eigenvalue $E(x)$, where $E(x)$ is the number of elements of $\{S_a\}$ that anticommute with $x$. Thermal fluctuations may excite the memory, but excitations with energy cost $E$ are suppressed by the Boltzmann factor $e^{-E/\tau}$ where $\tau$ is the temperature (and Boltzmann's constant $k_B$ has been set to one). Following \cite{BravyiTerhal2008no-go}, we suppose that the environment applies a sequence of weight-one Pauli operators to the system, so that the error history after $t$ steps can be described as a walk on the Pauli group, starting at the identity:
\[
\{x_i\in P, i = 0,1,2,3, \dots t\},
\]
where $x_0= I$, and $x_{i+1}x_i^{-1}$ has weight one. Let $\mathcal{P}(z)$ denote the set of all such walks, with any number of steps, that start at $I$ and terminate at $z\in P$. 
We define
\[
\Delta(z) \equiv \min_{\gamma\in \mathcal{P}(z)}\max_{x\in \gamma} E(x),
\]
the minimum energy barrier that must be surmounted by any walk that reaches Pauli operator $z$. Thus such walks occur with a probability per unit time suppressed by the Boltzmann factor $e^{-\Delta(z)/\tau}$. We also define
\begin{align*}
\Delta_{\rm min}& \equiv \min_{x\in S^\perp\setminus G}  \Delta(x),\\
\Delta_{\rm max}& \equiv \max_{x\in S^\perp} \min_{y\in G} \Delta(xy).
\end{align*}
Here $\Delta_{\rm min}$ is the lowest energy barrier protecting any nontrivial dressed logical operator (representing a nontrivial coset of $S^\perp/G$), and $\Delta_{\rm max}$ is the highest such energy barrier.

We say that a quantum memory is \emph{self correcting} if $\Delta_{\rm min}$ grows faster than logarithmically with $L$. In that case \emph{all} nontrivial logical operators are suppressed by a Boltzmann factor whose reciprocal grows super-polynomially with $L$. We say that the quantum memory is \emph{partially self correcting} if $\Delta_{\rm max}$ grows faster than logarithmically with $L$. In that case \emph{at least one} logical operator is protected by an energy barrier that increases with system size. Though the Pauli walk may not be a particularly accurate description of noise in realistic systems,  it allows us to define the notion of barrier height precisely, and to state the criteria for self correction and partial self correction simply. Furthermore, we expect the Boltzmann factor $e^{-\Delta/\tau}$ suppressing the Pauli walk to provide a reasonable (though crude) estimate of the logical error rate for more realistic noise models, assuming that the system attains thermal equilibrium.

Bravyi and Terhal \cite{BravyiTerhal2008no-go}, and Kay and Colbeck \cite{KayColbeck2008Quantum}, showed that no two-dimensional local subsystem code with local stabilizer generators can be self correcting. On the other hand, partially self-correcting quantum memories are certainly possible in two dimensions --- the Ising model, regarded as a quantum repetition code, is an example. In the Ising model, the logical bit flip operator flips every qubit, hence $\tilde d = L^2$. In the Pauli walk that reaches the logical bit flip and traverses the lowest barrier, a domain wall of length $\Omega(L)$ sweeps across the system; hence $\Delta_{\rm max}= \Omega(L)$.   Theorem \ref{thm:subsystem-tradeoff} shows that this high value of $\tilde d$ for the logical bit flip is possible only because the code distance $d$ is $O(1)$, and hence  a logical phase flip can be realized by an operator of constant weight.  

But suppose that, as in the toric code \cite{Kitaev2003Fault-tolerant}, a logical phase flip can occur only if a thermally activated localized quasiparticle propagates across the system. Thus $\tilde d = \Omega(L)$ for the logical phase flip. Can the logical bit flip still be protected by a high barrier? Arguing as in \cite{BravyiTerhal2008no-go}, and invoking Theorem \ref{thm:subsystem-tradeoff}, we see that under this condition robust protection against bit flips cannot be achieved using a local subsystem code with local stabilizer generators.

\begin{theorem}\emph{(Limitation on partial self correction in local subsystem codes)}
For a two-dimensional local subsystem code, with qubits residing at sites of an $L\times L$ square lattice, suppose that $\{S_a\}$ is a (possibly overcomplete) set of \emph{geometrically local} stabilizer generators, where the number of generators acting on each qubit is an $L$-independent constant. Consider a quantum memory governed by the Hamiltonian eq.\eqref{eq:Hamiltonian}. If the code distance is $d=\Omega(L)$, then the memory is not partially self correcting --- \emph{i.e.}, $\Delta_{\rm max} = O(1)$. More generally, if the code distance is $d= \Omega(L^\alpha)$ in $D$ spatial dimensions, then $\Delta_{\rm max} = O(L^\beta)$, where $\beta= D-1 - \alpha/(D-1)$.
\label{thm:no-partial}
\end{theorem}
\begin{proof}
Let $w$ be the interaction range of the gauge generators of the subsystem code and let $w_S$ be the interaction range of the stabilizer generators. 

For any dressed logical operator $x$ supported on this set, we may build a Pauli walk that starts at $I$ and ends at $x$ by first building the horizontal strings column by column and then building the vertical strings row by row. At each stage of this walk, any ``excited'' local stabilizer $S_a$ such that $S_a = -1$ acts only on qubits in a $w_S\times w_S$ square that contains qubits either at the boundary of the walk or in the intersection of a horizontal and vertical string. The number of such qubits is $O(1)$ and the total number of stabilizer generators acting on these qubits is $O(1)$. Therefore, the energy cost of the partially completed walk, and hence $\Delta_{\rm max}$, are $O(1)$.

In $D$ spatial dimensions, the proof of Theorem \ref{thm:subsystem-tradeoff} shows that the support of any dressed logical operator can be reduced to a network of overlapping $(D{-}1)$-dimensional slabs, where each slab has constant width and slabs with the same orientation are separated by distance $l$ such that $l^{D-1}=\Omega(d)$; hence $l=\Omega(L^{\alpha/(D-1)})$ if $d=\Omega(L^\alpha)$. For any dressed logical operator supported on this set of slabs, we may build a Pauli walk that sweeps across the system, such that at each stage of the walk the excited stabilizer generators are confined to a $(D{-}1)$-dimensional ``surface.'' This surface may be oriented such that it cuts across each slab on a $(D{-}2)$-dimensional surface with weight $O(L^{D-2})$. There are $O(L/l)$ such intersections; therefore during the walk the total number of excited stabilizer generators (and hence the energy cost) is $O((L/l)L^{D-2})=O(L^\beta)$, where $\beta = D-1 - \alpha/(D-1)$.
\end{proof}


We needed to assume that each $S_a$ is geometrically local to ensure that eq.\eqref{eq:Hamiltonian} is a geometrically local Hamiltonian. For any local subsystem code with geometrically local gauge generators, whether or not the stabilizer generators are also geometrically local,  the Hamiltonian \cite{Bacon2006Operator}
\[
H= -\sum_a \frac{1}{2}\lambda_a\left(G_a - I\right)
\]
is geometrically local, where now $\{G_a\}$ is the set of gauge generators. However, because the gauge generators are not mutually commuting, the energetics of a Pauli walk is not easy to study in this model, which is beyond the scope of Theorem \ref{thm:no-partial}.

\section{Are there self-correcting local commuting projector codes in two dimensions?}
For a two-dimensional local commuting projector code, the simultaneous eigenspace with eigenvalue one of the projectors $\{\Pi_a\}$, the code space is the degenerate ground state with energy $E=0$ of the Hamiltonian
\begin{equation}\label{eq:projector-hamiltonian}
H= -\sum_a\frac{1}{2}\left(\Pi_a - I\right).
\end{equation}
If only a constant number of projectors act on each qubit, then an operator supported on a set $M$ can increase the energy by at most $c|M|$, where $c$ is a constant. Since Theorem \ref{thm:slab} establishes the existence of a nontrivial logical operator supported on a narrow strip, one might anticipate that, by arguing as in the proof of Theorem \ref{thm:no-partial}, we can show that this system is not self correcting or partially self correcting.

We may envision a sequence of operations interpolating between the identity and a nontrivial logical operator, where each operation in the sequence could plausibly evolve from the previous operation due to the action of a thermal bath. In the strip $M$ of constant width that supports a nontrivial logical operator $\mathcal{O}$, we can divide the qubits into two subsets $A$ and $B=M\setminus A$, imagining that the interface between $A$ and $B$ gradually creeps along the strip. 

Now, however, we encounter an important distinction between stabilizer codes and more general commuting projector codes. For a stabilizer code, the nontrivial logical operator supported in $M$ can be chosen to be a Pauli operator, and hence the product of an operator supported in $A$ and an operator supported in $B$. For a commuting projector code, a logical operator supported in $M$ may actually be \emph{entangling} across the $A$-$B$ cut. Are we assured that this entangling operation can be built up gradually due to the effects of local noise?

We have not been able to settle this question. We {\em can} say that in any two-dimensional local commuting projector code there exists a nontrivial logical operator that is only \emph{slightly entangling} across any cut through the strip. This property, however, might not suffice to guarantee that the logical operator can be constructed as a product of physical operations, where each operation acts on a constant number of system qubits near the $A$-$B$ cut and also on a constant number of ancillary qubits in the ``environment.'' 

To define the notion of ``slightly entangling'' for an operator $\mathcal{O}$ supported on $AB$, we perform a Schmidt decomposition
\[
\mathcal{O}= \sum_\alpha \sqrt{\lambda_\alpha}~\mathcal{O}_A^\alpha\otimes \mathcal{O}_B^\alpha;
\]
here $\{\lambda_\alpha\}$ is a set of nonnegative real numbers, while $\{\mathcal{O}_A^\alpha\}$ is a set of operators supported on $A$ and $\{\mathcal{O}_B^\alpha\}$ is a set of operators supported on $B$, with the normalization conditions
\begin{align*}
&{\rm tr}\left(\mathcal{O}_A^{\alpha\dagger} \mathcal{O}_A^\beta\right)=2^{|A|}~\delta^{\alpha\beta},\nonumber\\
&{\rm tr}\left(\mathcal{O}_B^{\alpha\dagger} \mathcal{O}_B^\beta\right)=2^{|B|}~\delta^{\alpha\beta}.
\end{align*}
The number of nonzero terms in the Schmidt decomposition is the Schmidt rank of $\mathcal{O}$, and we say that $\mathcal{O}$ is slightly entangling if its Schmidt rank is a constant independent of system size. 

As we know from Theorem \ref{thm:slab}, for a two-dimensional local commuting projector code on an $L\times L$ lattice, there is a nontrivial logical operator supported on a vertical strip $M$ with dimensions $r\times L$, where $r$ is a constant. $M$ can be regarded as the disjoint union of an $r\times h$ rectangle $A$ covering the bottom of $M$ and an $r\times (L-h)$ rectangle $B$ covering the top of $M$. We can prove
\begin{lem}
\label{lem:e_slightly_entangling_lop}
\emph{(Existence of slightly entangling logical operators)}
For a nontrivial two-dimensional local commuting projector code, there is a nontrivial Hermitian logical operator $\mathcal{H}$ supported on a strip of constant width $M'$ such that, for any division of $M'$ into constant-width rectangles $A$ and $B$, $\mathcal{H}$ is slightly entangling across the $A$-$B$ cut. 
\end{lem}
\begin{proof}
We know from Lemmas \ref{lem:slab} and \ref{lem:projector-support} that there is a noncorrectable constant-width strip $M$ and a Pauli operator $P_M$ supported on $M$ such that 
\[
\mathcal H = \Pi_{M'} P_M \Pi_{M'},
\]
is a nontrivial Hermitian logical operator supported on $M'$; here $\Pi_{M'}= \prod_{\Pi_a\in \mathcal{S}'_M} \Pi_a$ and $\mathcal{S}'_M$ is the set of projectors that act nontrivially on $M$.
The Pauli operator $P_M$ is a product operator, with Schmidt number one across the $A$-$B$ cut. Among the local projectors occurring in the product $\Pi_{M'}$, those fully supported on either $A$ or $B$ have no effect on the Schmidt number of $\Pi_{M'} P_M \Pi_{M'}$, and only a constant number of the projectors act nontrivially on both $A$ and $B$. Since each such $\Pi_a$ is supported on a constant number of qubits, the action of $\Pi_a$ increases the Schmidt number by a constant.
Thus $\mathcal{H}$ has constant Schmidt number, \emph{i.e.} is slightly entangling.
\end{proof}

We may relax the notion of slightly entangling, regarding an operator $\mathcal{O}$ as slightly entangling if it may be {\em well approximated} by an operator with constant Schmidt rank. In this sense the unitary logical operator $U = \exp( i \lambda \mathcal H)$ is also slightly entangling. We may expand the exponential as a power series where each term has a Schmidt rank independent of system size; furthermore, the power series expansion truncated at constant order approximates the exponential function very well with respect to the operator norm.

Now we might hope to construct a slightly entangling logical operator $\mathcal{O}$, supported on a constant-width vertical strip, by gradually building its support one horizontal row of qubits at a time. However, Lamata {\em et al.} \cite{Lamata2008Sequential} showed that, if $\mathcal{O}$ is entangling, then it cannot be obtained as a product of \emph{unitary} operators where each of these operators acts on just a few rows of system qubits and on a shared ancillary system. 

An alternative procedure for gradually building a nontrivial logical error has been proposed by Landon-Cardinal and Poulin \cite{Poulin2012Unpublished}. They envision a walk along the strip such that, in each step of the walk, first a constant size set of qubits depolarizes, and then the code projectors acting on that set are applied. If the projection fails to accept the state, the step can be repeated until the projection succeeds. 

This procedure could fail if at some stage the projection succeeds with zero probability. But Landon-Cardinal and Poulin \cite{Poulin2012Unpublished} have shown that their procedure eventually generates a nontrivial logical error (and hence that the code is not self correcting) for any local commuting projector code obeying a ``local topological order'' criterion \cite{Hastings2010Short}. Whether self-correcting two-dimensional local commuting projector codes are possible remains open, though, because topologically ordered codes that violate {\em local} topological order have not been ruled out. 

%
%

\section{Conclusion}
The quantum accuracy threshold theorem \cite{Gottesman2009Introduction} shows that quantum information can be reliably stored and processed by a noisy physical system if the noise is not too strong. But can quantum information be protected ``passively'' in a macroscopic physical system governed by a static local Hamiltonian, at a sufficiently low nonzero temperature? This question \cite{DennisKitaevLandahlEtAl2002Topological,Bacon2006Operator}, aside from its far-reaching potential implications for future quantum technologies, is also a fundamental issue in quantum many-body physics. Hamiltonians derived from local quantum codes, whose properties are relatively easy to discern, can provide us with valuable insights.

A two-dimensional ferromagnet can be a self-correcting classical memory, but a Hamiltonian based on a two-dimensional local subsystem code with local stabilizer generators cannot be a self-correcting quantum memory \cite{BravyiTerhal2008no-go,KayColbeck2008Quantum}.
We have shown that for two-dimensional local subsystem code with local stabilizer generators on an $L \times L$ square lattice,
robust \emph{classical} protection is impossible if the code distance is $d=\Omega(L)$, as expected for a topologically ordered two-dimensional system.
More generally, we have studied how the code distance $d$ limits the size of the support of arbitrary nontrivial logical operators, in both local subsystem codes and local commuting projector codes. In view of the upper bound $d=O(L^{D-1})$ on the code distance, we may write $d= \Theta (L^{(D-1)(1 - \delta)})$ where $0\le \delta \le 1$, and thus our upper bound eq.\eqref{eq:main-result} on the weight of logical operators becomes
\[
\tilde d = O(L^{D-1+\delta}).
\]
In particular, in three dimensions, $d=\Omega(L)$ implies $\tilde d=O(L^{5/2})$. We have also shown that any two-dimensional local commuting projector code admits a nontrivial logical string operator which is only slightly entangling across any cut through the string. 

Our arguments modestly extend the findings of \cite{BravyiPoulinTerhal2010Tradeoffs,BravyiTerhal2008no-go,Bravyi2010Subsystem,KayColbeck2008Quantum}, and use similar ideas. In passing, we also proved a Cleaning Lemma for subsystem codes
based on ideas from \cite{YoshidaChuang2010Framework}, proved a Cleaning Lemma for local commuting projector codes.
Our methods might find further applications in future studies of quantum memories based on local codes. 

\acknowledgments
We are grateful to Salman Beigi, Alexei Kitaev, Robert K\"onig, Olivier Landon-Cardinal, and Norbert Schuch for helpful discussions, and we especially thank David Poulin for useful comments on the manuscript. 
This research was supported in part by NSF under Grant No. PHY-0803371, by DOE under Grant No. DE-FG03-92-ER40701, by NSA/ARO under Grant No. W911NF-09-1-0442, and by the Korea Foundation for Advanced Studies. The Institute for Quantum Information and Matter (IQIM) is an NSF Physics Frontiers Center with support from the Gordon and Betty Moore Foundation. 

\appendix
\section{Holographic lemma for local stabilizer codes}
\label{app:holographic_lemma_stabilizer_codes}

We say that a local stabilizer code has interaction range $w$ if each stabilizer generator has support on a hypercube containing $w^D$ sites. For this case, we can improve the criterion for correctability of a hypercube, found for local subsystem codes in Lemma \ref{lem:subsystem-hypercube}.

\begin{lem}\emph{(Expansion Lemma for local stabilizer codes)}
For a local stabilizer code, suppose that $\partial_+M$, $A$, and $M\setminus A$  are all correctable, where $\partial_- M \subseteq A \subseteq M$. Then $M$ is also correctable. 
\end{lem}
\begin{proof}
Suppose, contrary to the claim, that there is a nontrivial logical operator $x$ supported on $M$. Then, because $A$ is correctable, Lemma \ref{lem:clean-region} implies that there is a stabilizer generator $y$ such that $xy$ acts trivially on $A$. Furthermore, $y$ can be expressed as a product of local stabilizer generators, each supported on $M'=M\cup\partial_+M$. Thus $xy$ is a product of two factors, one supported on $M\setminus A$ and the other supported on $\partial_+M$. Because $\partial_-M\subseteq A$, no local stabilizer generator acts nontrivially on both $M\setminus A$ and $\partial_+M$; therefore, each factor commutes with all stabilizer generators and hence is a logical operator. Because $M\setminus A$ and $\partial_+M$ are both correctable, each factor is a trivial logical operator and therefore $xy$ is also trivial. It follows that $x$ is trivial, a contradiction. 
\end{proof}

Now, if the interaction range is $w$ and $M$ is a hypercube with linear size $l$, we choose $A$ so that $M\setminus A$ is a hypercube with linear size $l-2(w-1)$, and we notice that $\partial_+M$ is contained in a hypercube with linear size $l+2(w-1)$. Thus both $M\setminus A$ and $\partial_+ M$ are correctable provided that
\begin{align*}
|\partial_+ M| &\le \left[l +2(w-1)\right]^D - l^D \nonumber\\
&\le 2(w-1)D\left[l +2(w-1)\right]^{D-1} < d.
\end{align*}
Reasoning as in the proof of Lemma \ref{lem:subsystem-hypercube}, we conclude that:

\begin{lem}\emph{(Holographic Principle for local stabilizer codes)}
\label{lem:stabilizer-hypercube}
For a $D$-dimensional local stabilizer code with interaction range $w>1$ and distance $d>1 $, a hypercube with linear size $l$ is correctable if 
\begin{equation}\label{eq:hypercube-size-stabilizer}
2(w-1)D\left[l+ 2(w-1)\right]^{D-1} < d.
\end{equation}
\end{lem}

\noindent To ensure that the hypercube $M$ is correctable, it suffices for its $(w-1)$-thickened boundary, rather than its $\left[2(w-1)\right]$-thickened boundary, to be smaller than the code distance. 

\section{A noncorrectable set that supports no nontrivial logical operator}
\label{app:counter_example_nolop_noncorrectable}
Here we give a simple example illustrating that for some quantum codes a noncorrectable set need not support a nontrivial logical operator. For $n=2$ qubits, consider the three-dimensional code space spanned by the orthogonal vectors
\begin{align*}
&|\phi\rangle =\frac{1}{\sqrt{2}}\left(|00\rangle + |11\rangle\right),\\
&|\psi\rangle =|01\rangle,\\
&|\chi\rangle =|10\rangle;
\end{align*}
this is the eigenspace with eigenvalue 1 of the projector
\[
\Pi = |\phi\rangle\langle \phi | + |\psi\rangle\langle \psi| +|\chi\rangle\langle \chi|.
\]
If the first qubit is mapped to $|0\rangle$, then $|\phi\rangle$ is no longer perfectly distinguishable from $|\psi\rangle$ or $|\chi\rangle$; hence erasure of this qubit is not correctable. (Similarly, the second qubit is also a noncorrectable set.)

Is there a logical operator supported on the first qubit? Suppose that 
\[
L = \begin{pmatrix} a & b \\ c & d \end{pmatrix}
\]
is an operator acting on the first qubit. Then $L\otimes I |\psi\rangle = a|01\rangle + c |11\rangle$ is a code vector only if $c=0$, and  $L\otimes I |\chi\rangle = b|00\rangle + d|10\rangle$ is a code vector only if $b=0$. Furthermore, if $b=c=0$, then $L\otimes I |\phi\rangle = \left(a|00\rangle + d|11\rangle\right)/\sqrt{2}$ is a code vector only if $a=d$. Thus $L$ is a multiple of the identity, a trivial operator.

\end{document}